\documentclass[10pt]{article}
\usepackage{amsmath}
\usepackage{amsthm}
\usepackage{amssymb}
\usepackage[english]{babel}
\usepackage[dvips]{graphicx}

\newtheorem{theorem}{Theorem}

\newcommand{\bs}{\boldsymbol}
\newcommand{\ri}{\mathrm{i}\,}

\begin{document}

\begin{flushright}
\footnotesize

Mech. Tverd. Tela

No. 35 (2005), pp. 38-48

(In Russian)

\end{flushright}

\normalsize

\begin{center}

{\Large \bf Bifurcation Diagram \\
of the Generalized 4th Appelrot Class}

\end{center}

\noindent UDC 531.38

\begin{center}
{\bf M. P. KHARLAMOV}

{\footnotesize Presented October 1, 2005}

\end{center}

{\footnotesize The article continues the author's publication in
[Mech. Tverd. Tela, No. 34, 2004], in which the generalizations of
the Appelrot classes of the Kowalevski top motions are found for the
case of the double force field. We consider the analogue of the 4th
Appelrot class. The trajectories of this family fill the surface
which is four-dimensional in the neighborhood of its generic points.
The complete system of two integrals is pointed out. For these
integrals the bifurcation diagram is established and the admissible
region for the corresponding constants is found. }

\section{Introduction}
Consider a rigid body with a fixed point~$O$. Let the principal
moments of inertia at $O$ satisfy the ratio $2:2:1$. Suppose that
the moment of external forces with respect to $O$ has the form
$$
{\bf{e}}_1  \times {\bs{\alpha }} + {\bf{e}}_2  \times {\bs{\beta
}},
$$
where the vectors ${\bf{e}}_1 ,{\bf{e}}_2 $ are fixed in the body
and parallel to the equatorial plane of the inertia ellipsoid, and
${\bs{\alpha }},{\bs{\beta }}$ are the vectors constant in the
inertial space.

It is shown in the work \cite{bib1} that without loss of generality
one can consider ${\bf{e}}_1 ,{\bf{e}}_2 $ to form the orthonormal
pair (in particular, to be the principal inertia unit vectors), and
${\bs{\alpha }},{\bs{\beta }}$ to be mutually orthogonal. Let ${\bf
e}_3 = {\bf e}_1 \times {\bf e}_2 $. Choose $O{\bf e}_1 {\bf e}_2
{\bf e}_3$ as the moving frame. Denote by ${\bs \omega}$ the angular
velocity vector. In the dimension less variables the rotation of the
body is described by the Euler--Poisson equations
\begin{equation}\label{kh1}
\begin{array}{c}
2\dot \omega _1   = \omega _2 \omega _3  + \beta _3 ,\;\quad
2\dot\omega _2 =  - \omega _1 \omega _3  - \alpha _3 ,\; \quad \dot \omega _3 = \alpha _2  - \beta _1 , \\[2mm]
\dot\alpha _1 = \alpha _2 \omega _3  - \alpha _3 \omega _2 ,\; \quad
\dot\beta _1 = \beta _2 \omega _3  - \beta _3 \omega _2 \qquad
(123)\ .
\end{array}
\end{equation}
The (123) symbol means that the remaining equations of the Poisson
group are obtained by the cyclic substitution of the indexes.

The phase space $P^6 $ of system (\ref{kh1}) is defined in
${\bf{R}}^9 ({\bs{\omega }},{\bs{\alpha }},{\bs{\beta }})$ by the
geometric integrals
\begin{equation}\label{kh2}
\alpha _1^2  + \alpha _2^2  + \alpha _3^2  = a^2 ,\quad \beta _1^2
+ \beta _2^2  + \beta _3^2  = b^2 ,\quad \alpha _1 \beta _1  +
\alpha _2 \beta _2  + \alpha _3 \beta _3  = 0.
\end{equation}
We suppose that
\begin{equation}\label{kh3}
a > b > 0.
\end{equation}
Then system (\ref{kh1}), (\ref{kh2}) does not have any cyclic
integrals and is not reducible, by the standard procedure, to a
Hamiltonian system with two degrees of freedom. Nevertheless, it is
completely integrable due to the existence of the first integrals in
involution
\begin{equation}\label{kh4}
\begin{array}{l}
H = \omega _1^2  + \omega _2^2  + {1 \over 2}\omega _3^2  - (\alpha _1  + \beta _2 ),
\\[2mm]
K = (\omega _1^2  - \omega _2^2  + \alpha _1  - \beta _2 )^2  +
(2\omega _1 \omega _2  + \alpha _2  + \beta _1 )^2 , \\[2mm]
G = (\alpha _1 \omega _1  + \alpha _2 \omega _2  + {1 \over 2}\alpha _3 \omega _3 )^2  +
(\beta _1 \omega _1  + \beta _2 \omega _2  + {1
\over 2}\beta _3 \omega _3 )^2  +  \\[2mm]
\qquad {} + \omega _3 (\gamma _1 \omega _1  + \gamma _2 \omega _2 + {1 \over 2}\gamma _3
\omega _3 ) - \alpha _1 b^2  - \beta _2 a^2
\end{array}
\end{equation}
(here $\gamma _i $ stand for the components of the vector
${\bs{\alpha }} \times {\bs{\beta }}$ immovable in space).

The integral $K$ was first shown by O.I.\,Bogoyavlensky~\cite{bib2},
and the integral $G$ (in more general form for a gyrostat) was found
by A.G.\,Reyman and M.A.\,Semenov-Tian-Shansky~\cite{bib3}.

In the work \cite{bib1} the set of critical points is found for the
integral map
\begin{equation}\label{kh5}
H \times K \times G:P^6  \to {\bf{R}}^3.
\end{equation}
It is shown that this set is the union of three sets $\mathfrak{M},
\mathfrak{N}, \mathfrak{O}$, which are almost everywhere the smooth
four-dimensional submanifolds in $P^6 $ and in the neighborhood of
the generic points are defined by two invariant relations.

The first critical set $\mathfrak{M}$ was found in the work
\cite{bib2}. It coincides with the zero level of the integral $K$
and generalizes the $1^{\mathrm{st}}$ Appelrot class in the
Kowalevski problem \cite{bib4}. The phase topology of the dynamical
system induced on $\mathfrak{M}$ was studied by
D.B.\,Zotev~\cite{bib5}.

The motions on the manifold $\mathfrak{N}$ found in \cite{bib6} are
investigated in \cite{bib8, bib7}. It is shown that this family of
motions is the generalization of the so-called {\it especially
remarkable} motions of the $2^{\mathrm{nd}}$ and $3^{\mathrm{rd}}$
Appelrot classes. The bifurcation diagram for the pair of almost
everywhere independent first integrals on $\mathfrak{N}$ is
constructed in \cite{bib7}, the equations on $\mathfrak{N}$ are
separated, the bifurcations of the Liouville tori are studied in
\cite{bib8}.

The present work is devoted to the investigation of some properties
of system (\ref{kh1}) restricted to the invariant subset
$\mathfrak{O}$ given in $P^6 $ by the system of the invariant
relations \cite{bib1}
\begin{equation}\label{kh6}
R_1  = 0,\;R_2  = 0,
\end{equation}
where
$$
\begin{array}{l}
R_1  = (\alpha _3 \omega _2  - \beta _3 \omega _1 )\omega _3  -
2\beta _1 \omega _1^2  + 2(\alpha _1  - \beta _2 )\omega _1 \omega
_2  + 2\alpha _2 \omega _2^2 , \\[2mm]
R_2  = (\alpha _3 \omega _1  + \beta _3 \omega _2 )\omega _3^2  + [\alpha _3^2  + \beta
_3^2  + 2\alpha _1 \omega _1^2  + 2(\alpha _2 + \beta _1 )\omega _1 \omega _2  + 2\beta
_2 \omega _2^2 ]\omega _3
+  \\[2mm]
\qquad {}+ 2\alpha _3 [(\alpha _1  - \beta _2 )\omega _1  + (\alpha _2  + \beta _1
)\omega _2 ] + 2\beta _3 [(\alpha _2  + \beta _1 )\omega _1 - (\alpha _1  - \beta _2
)\omega _2 ].
\end{array}
$$

\section{Partial integrals}
Note that at the points
\begin{equation}\label{kh7}
\omega _1  = \omega _2  = 0,\quad \alpha _3  = \beta _3  = 0
\end{equation}
equations (\ref{kh6}) are dependent on $\mathfrak{O}$. If we assume
that equations (\ref{kh7}) hold during some time interval (and then
identically in $t$ along the whole trajectory), then we come to the
family of pendulum type motions
\begin{equation}\label{kh8}
\begin{array}{c}
{\bs{\alpha }} = a({\bf{e}}_1 \cos \theta  - {\bf{e}}_2 \sin \theta ),\quad {\bs{\beta }}
= \pm b({\bf{e}}_1 \sin \theta  + {\bf{e}}_2 \cos
\theta ),\quad {\bs \alpha}\times{\bs \beta} \equiv  \pm ab{\bf{e}}_3 , \\[2mm]
{\bs{\omega }} = \dot \theta  {\bf{e}}_3 ,\quad \ddot \theta = - (a
\pm b)\sin \theta
\end{array}
\end{equation}
noticed in the work \cite{bib1}. The constants of integrals
(\ref{kh4}) at such trajectories satisfy one of the following:
\begin{equation}\label{kh9}
g = abh,\quad k = (a - b)^2 ,\quad h \geqslant  - (a + b)
\end{equation}
or
\begin{equation}\label{kh10}
g =  - abh,\quad k = (a + b)^2 ,\quad h \geqslant  - (a - b).
\end{equation}
Denote by $\Omega$ the set of points belonging to trajectories
(\ref{kh8}). Let $\mathfrak{O}^ * =\mathfrak{O} \backslash \Omega $.

Recall that in the classical case of S.\,Kowalevski (${\bs{\beta }}
= 0$) there exists the area integral. Traditionally it is
represented with the one-half multiplier
\begin{equation}\label{kh11}
L = \frac {1}{2} {\bf I}{\bs \omega}{\bs \cdot} {\bs \alpha}.
\end{equation}
Then provided that ${\bs{\beta }} = 0$ the integral $G$ turns into
$L^2$.

Let $\ell $ be the constant of integral (\ref{kh11}). According to
G.G.\,Appelrot's classification the $4^{\mathrm{th}}$ class of the
{\it especially remarkable} motions is defined by the following
conditions:

1) the second polynomial of Kowalevski has a multiple root, one of
the Kowalevski variables remains constant and equal to the multiple
root $s$ of the corresponding Euler resolvent defined as ${\varphi
(s) = s(s - h)^2 + (a^2 - k)s - 2\ell ^2}$:
\begin{equation}\label{kh12}
\varphi (s) = 0,\quad \varphi '(s) = 0;
\end{equation}

2) the first two components of the angular velocity are constant and
equal to
\begin{equation}\label{kh13}
\omega _1  =  - {\ell  \over s},\quad \omega _2  = 0.
\end{equation}

The next statement establishes the analogue of conditions
(\ref{kh13}) for the generalized top.

\begin{theorem} For any trajectory in the set $\mathfrak{O}^*$ the values
$\displaystyle{\frac {{\bf I}{\bs \omega}{\bs \cdot}{\bs \alpha}} {
{\bf I}{\bs \omega}{\bs \cdot}{\bf e}_1 }}$ and $\displaystyle{\frac
{ {\bf I}{\bs \omega}{\bs \cdot}{\bs \beta}} { {\bf I}{\bs
\omega}{\bs \cdot}{\bf e}_2 }}$ are equal to each other and
constant.
\end{theorem}
\begin{proof} Denote
$$
\begin{array}{c}
{\bf{M}} = {\bf{I}\bs{\omega}}, \\[1mm]
M_1  = {\bf{I}\bs{\omega}}{\bs \cdot}{\bf{e}}_1  = 2\omega _1 ,\quad M_2
 = {\bf{I}\bs{\omega}}{\bs \cdot}{\bf{e}}_2  = 2\omega _2 , \\[1mm]
M_\alpha   = {\bf{I}\bs{\omega}}{\bs \cdot}{\bs{\alpha }} = 2\alpha _1 \omega _1  +
2\alpha _2 \omega _2  + \alpha _3 \omega _3 ,\quad M_\beta   = {\bf{I}\bs{\omega}}{\bs
\cdot}{\bs{\beta }} = 2\beta _1 \omega _1  + 2\beta _2 \omega _2  + \beta _3 \omega _3 .
\end{array}
$$
The first equation (\ref{kh6}) becomes
\begin{equation}\label{kh14}
{{M_\alpha  } \over {M_1 }} - {{M_\beta  } \over {M_2 }} = 0.
\end{equation}
Introduce the function
$$
S =  - {{M_\alpha  M_1  + M_\beta  M_2 } \over {M_1^2  + M_2^2 }}.
$$
Its derivative in virtue of equations (\ref{kh1}) is
$$
{{dS} \over {dt}} = {1 \over {4(M_1^2  + M_2^2 )^2 }}[(M_1^2  +
M_2^2 )\omega _3  + 4\alpha _3 M_1  + 4\beta _3 M_2 ](M_\beta  M_1
- M_\alpha  M_2 ).
$$
Then (\ref{kh14}) implies that the right hand part vanishes
identically. Therefore, $S$ is the partial integral on the set
$\mathfrak{O}^*$ . Let $s$ stand for the constant of this integral,
\begin{equation}\label{kh15}
{{M_\alpha  M_1  + M_\beta  M_2 } \over {M_1^2  + M_2^2 }} =  - s.
\end{equation}

From (\ref{kh14}), (\ref{kh15}) we obtain that $ M_\alpha = - sM_1$,
$M_\beta =  - s M_2$ with the constant value $s$.
\end{proof}

{\bf Remark 1}. In virtue of condition (\ref{kh14}) the function $S$
can be also written in the form
\begin{equation}\label{kh16}
S =  - {1 \over 2}({{M_\alpha  } \over {M_1 }} + {{M_\beta  } \over
{M_2 }}).
\end{equation}

{\bf Remark 2}. Note the interesting geometric feature of the
kinetic momentum vector motion on the trajectories considered.
Introduce the immovable orthonormal basis in the $O{\bs{\alpha \beta
}}$-plane
$$
{\bs{\nu }}_1  = {{\bs{\alpha }} \over a},\quad {\bs{\nu }}_2  =
{{\bs{\beta }} \over b}.
$$
Let $m_1  = {\bf{M}}{\bs \cdot}{\bs{\nu }}_1 ,\;m_2  = {\bf{M}}{\bs
\cdot}{\bs{\nu }}_2$. Then $M_\alpha   = a m_1 ,\;M_\beta   = b m_2
$, and condition (\ref{kh14}) yields
$$
{{M_2 } \over {M_1 }} = {b \over a}\,{{m_2 } \over {m_1 }},
$$
or
$$
\mathop{\rm tg}\nolimits \vartheta  = {b \over a}\,\mathop{\rm
tg}\nolimits \vartheta _0 ,
$$
where $\vartheta ,\vartheta _0 $ are the polar angles of the
projections of the vector ${\bf{M}}$, respectively, onto the
equatorial plane of the body and onto the plane of the direction
vectors of the forces fields.

\begin{theorem} On the set $\mathfrak{O}$  system $(\ref{kh1})$
has the partial integral
\begin{equation}\label{kh17}
\begin{array}{l}
{\rm T} = (\alpha _3 \omega _1  + \beta _3 \omega _2 )\omega _3  + 2\alpha _1 \omega _1^2
+ 2(\alpha _2  + \beta _1 )\omega _1 \omega _2  + 2\beta _2 \omega _2^2 - \\[2mm]
        \qquad{} - 2(\alpha _1 \beta _2  - \alpha _2 \beta _1 ) + a^2  + b^2 .
\end{array}
\end{equation}
\end{theorem}
\begin{proof} The derivative of (\ref{kh17}) in virtue of system
(\ref{kh1}) is
$$
{{d{\rm T}} \over {dt}} = {1 \over 4}\omega _3 R_1,
$$
and vanishes identically on $\mathfrak{O}$.
\end{proof}

Denote by $\tau $ the constant of the integral ${\rm T}$.

In the work \cite{bib1} equations (\ref{kh6}) are obtained from the
condition that the function with Lagrange's multipliers $s,\tau$
$$
2G + (\tau  - p^2 )H + sK
$$
has a critical point. Comparing (\ref{kh16}), (\ref{kh17}) with the
expressions for $s,\tau $ in \cite{bib1}, we see that these
multipliers are the constants of the above given integrals $S,{\rm
T}$.

According to (\ref{kh3}) introduce the positive parameters $p,r$ as
follows
$$
p^2  = a^2  + b^2 ,\quad r^2  = a^2  - b^2 .
$$
Let $h,k,g$ be the constants of the general integrals (\ref{kh4}).
Then equations (50) of the work \cite{bib1} give the following
relations on the set $\mathfrak{O}^*$,
\begin{equation}\label{kh18}
h = s + \frac {p^2 -\tau}{2s}, \; k = \tau  + \frac{\tau^2 -
2p^2\tau + r^4} {4s^2}, \; g = \frac{1}{2}(p^2 -\tau)s +
\frac{p^4-r^4} {4s}.
\end{equation}
These relations also can be considered as the parametric equations
of the sheet of the bifurcation diagram of the map (\ref{kh5}).
Eliminating of $\tau $ leads to the equations
\begin{equation}\label{kh19}
\psi (s) = 0,\quad \psi '(s) = 0,
\end{equation}
where
$$
\psi (s) = s^2 (s - h)^2  + (p^2  - k)s^2  - 2gs + {{p^4  - r^4 }
\over 4}.
$$
Under the condition ${\bs{\beta }} = 0$ ($p^2  = r^2  = a^2 $) we
have $\psi (s) = s\varphi (s)$. Thus, relations (\ref{kh19}) are
similar to conditions (\ref{kh12}). Therefore, the family of
trajectories on the set $\mathfrak{O}^*$ generalizes the family of
the especially remarkable motions of the $4^{\mathrm{th}}$ Appelrot
class.

\section{The equations of integral manifolds}
According to (\ref{kh18}) the finctions $S,{\rm T}$ form the
complete set of the first integrals on $\mathfrak{O}^*$. In
particular, the system of equations defining any integral manifold
$\{\zeta \in P^6 :H(\zeta ) = h,K(\zeta ) = k,G(\zeta ) = g\}$ is
now replaced by invariant relations (\ref{kh6}) and the equations
\begin{equation}\label{kh20}
S = s,\; {\rm T} = \tau.
\end{equation}

Introduce the complex change of variables~\cite{bib6} generalizing
the Kowalevski change for the top in the gravity field~\cite{bib9}
($\ri^2=-1$)
\begin{equation}\label{kh21}
\begin{array}{c}
x_1 = (\alpha_1  - \beta_2) + \ri(\alpha_2  + \beta_1),\quad
x_2 = (\alpha_1  - \beta_2) - \ri(\alpha_2  + \beta_1 ), \\
y_1 = (\alpha_1  + \beta_2) + \ri(\alpha_2  - \beta_1), \quad y_2 =
(\alpha_1  + \beta_2) -
\ri(\alpha_2  - \beta_1), \\
 z_1 = \alpha_3  + \ri\beta_3, \quad
z_2 = \alpha_3  - \ri\beta_3,\\
w_1 = \omega_1  + \ri\omega_2 , \quad w_2 = \omega_1  - \ri\omega_2,
\quad w_3 = \omega_3 \,.
\end{array}
\end{equation}

In variables (\ref{kh21}), system (\ref{kh6}), (\ref{kh20}) can be
presented in the form
\begin{equation}\label{kh22}
\begin{array}{l}
(y_2  + 2s)w_1  + x_1 w_2  + z_1 w_3  = 0,  \\
x_2 w_1  + (y_1  + 2s)w_2  + z_2 w_3  = 0,  \\
x_2 z_1 w_1  + x_1 z_2 w_2  + (\tau  - x_1 x_2 )w_3  = 0,  \\
2sw_1 w_2  - (x_1 x_2  + z_1 z_2 ) + \tau  = 0. \end{array}
\end{equation}
These equations must be added by geometrical integrals (\ref{kh2}),
which in variables (\ref{kh21}) can be written as follows
\begin{equation}\label{kh23}
z_1^2  + x_1 y_2  = r^2 ,\quad z_2^2  + x_2 y_1  = r^2 ,\quad x_1
x_2 + y_1 y_2  + 2z_1 z_2  = 2p^2.
\end{equation}

The space of variables (\ref{kh21}) has dimension 9 regarding the
fact that the following pairs must be complex conjugate $x_2 =
\overline{x_1}$, $y_2  = \overline{y_1}$, $z_2 = \overline{z_1}$,
and that $w_3$ is real. Seven relations (\ref{kh22}), (\ref{kh23})
define then the integral manifold. In the case when the the
integrals $S,{\rm T}$ are independent on this manifold it consists
of two-dimensional tori bearing quasi-periodic motions.

\section{Bifurcation diagram}
Introduce the integral map $J$ of the the dynamical system induced
on the closure of the set $\mathfrak{O}^*$,
$$
J(\zeta ) = (S(\zeta ),{\rm T}(\zeta )) \in {\bf{R}}^2, \quad \zeta
\in \mathop{\rm Cl}\nolimits (\mathfrak{O}^*).
$$
Due to the obvious compact character of the inverse images of the
points of ${\bf{R}}^2 $ the bifurcation diagram $\Sigma $ of the map
$J$ coincides with the set of its critical values.

\begin{theorem} The bifurcation diagram of the map
\begin{equation}\label{kh24}
J = S \times {\rm T}:\,\mathop{\rm Cl}\nolimits (\mathfrak{O}^ *  )
\to {\bf{R}}^2
\end{equation}
consists of the following subsets of the $(s,\tau)$-plane:

$\;1^ \circ)\; \tau  = (a + b)^2 ,\;s \in [ - a,0) \cup [b, + \infty
);$

$\;2^ \circ)\; \tau  = (a - b)^2 ,\;s \in [ - a, - b] \cup (0, +
\infty );$

$\;3^ \circ)\; s =  - a,\;\tau  \geqslant (a - b)^2 ;$

$\;4^ \circ)\; s =  - b,\;\tau  \geqslant (a - b)^2 ;$

$\;5^ \circ)\; s = b,\;\tau  \leqslant (a + b)^2 ;$

$\;6^ \circ)\; s = a,\;\tau  \leqslant (a + b)^2 ;$

$\;7^ \circ)\; \tau  = 0,\;s \in (0, + \infty );$

$\;8^ \circ)\; \tau  = a^2  + b^2  - 2s^2  + 2\sqrt {\mathstrut (a^2
- s^2 )(b^2 - s^2 )} ,\;s \in [ - b,0);$

$\;9^ \circ)\; \tau  = a^2  + b^2  - 2s^2  - 2\sqrt {\mathstrut(a^2
- s^2 )(b^2 - s^2 )} ,\;s \in (0,b];$

$10^ \circ)\; \tau  = a^2  + b^2  - 2s^2  + 2\sqrt {\mathstrut(s^2
- a^2 )(s^2 - b^2 )} ,\;s \in [a, + \infty ).$
\end{theorem}
\begin{proof}
Any point of dependence of equations (\ref{kh6}) are considered
critical for the map (\ref{kh24}) by definition. Take the points of
trajectories (\ref{kh8}) that belong to the closure of the set
$\mathfrak{O}^*$ and calculate the corresponding values of $J$.
These values must then be included in the bifurcation diagram. We
obtain those values $(s,\tau )$ for which equations (\ref{kh18})
give (\ref{kh9}), (\ref{kh10}). It is easily checked that the
half-line (\ref{kh10}) completely belongs to the surface given by
(\ref{kh18}). It corresponds to the case $1^\circ$. Consider the
half-line defined by (\ref{kh9}). The points of it lie on
(\ref{kh18}) only if $h^2 \geqslant 4ab$. The corresponding set is
given by $2^\circ$. The segment of (\ref{kh9}) in the limits
$$
- 2\sqrt {ab}  < h < 2\sqrt {ab}
$$
is the one-dimensional part of the bifurcation diagram of the map
(\ref{kh5}). For corresponding trajectories (\ref{kh8}) the value
$s$ is not defined. It means that such trajectories are isolated
from the set $\mathfrak{O}^*$. The isolated points in the
bifurcation diagrams of the reduced systems or of the systems
restricted to iso-energetic surfaces were met before only in the
Clebsch and Lagrange cases.

To find critical motions in $\mathfrak{O}^*$ use system
(\ref{kh22}), (\ref{kh23}). Introduce the variables $x,z$,
\begin{equation}\label{kh25}
x^2  = x_1 x_2 ,\quad z^2  = z_1 z_2 .
\end{equation}
It follows from the last equation (\ref{kh23}) that
\begin{equation}\label{kh26}
y_1 y_2  = 2p^2  - x^2  - 2z^2 ,
\end{equation}
and the first two equations give
\begin{equation}\label{kh27}
\begin{array}{l}
(z_1  + z_2 )^2  = 2r^2  - (x_1 y_2  + x_2 y_1 ) + 2z^2 ,  \\
(z_1  - z_2 )^2  = 2r^2  - (x_1 y_2  + x_2 y_1 ) - 2z^2 .
\end{array}
\end{equation}
Eliminate $z_1 ,z_2 $ in (\ref{kh23}):
$$
(r^2  - x_1 y_2 )(r^2  - x_2 y_1 ) = z^4 .
$$
Then using (\ref{kh26}) we obtain
\begin{equation}\label{kh28}
r^2 (x_1 y_2  + x_2 y_1 ) = r^4  + 2p^2 x^2  - (x^2  + z^2 )^2 .
\end{equation}
Denote
$$
\Phi _ \pm  (x,z) = (x^2  + z^2  \pm r^2 )^2  - 2(p^2  \pm r^2 )x^2.
$$
From (\ref{kh27}), (\ref{kh28}) we have
$$
r^2 (z_1  + z_2 )^2  = \Phi _ +  (x,z),\quad r^2 (z_1  - z_2 )^2  =
\Phi _ -  (x,z).
$$
Therefore, $x,z$ satisfy the inequalities
\begin{equation}\label{kh29}
\Phi _ +  (x,z) \geqslant 0,\quad \Phi _ -  (x,z) \leqslant 0.
\end{equation}

Notice that the equilibria of system (\ref{kh1}) are included in the
family of motions $\Omega $. In all other cases the determinant of
the first three equations (\ref{kh22}) in ${w_i \;(i = 1,2,3)}$ is
identically zero. Eliminating $z_1^2$, $z_2^2$, and $y_1 y_2 $ with
the help of (\ref{kh23}), (\ref{kh26}), we obtain
\begin{equation}\label{kh30}
\begin{array}{l}
2s[(r^2 x_1  - \tau y_1 ) + (r^2 x_2  - \tau y_2 )] = - r^2 (x_1 y_2  + x_2 y_1 ) +\\
\quad  + 2[2s^2 (\tau  - x^2 ) + p^2 (\tau  + x^2 ) - \tau (x^2  +
z^2 )].
\end{array}
\end{equation}
On the other hand, (\ref{kh25}) and (\ref{kh26}) yield
\begin{equation}\label{kh31}
(r^2 x_1  - \tau y_1 )(r^2 x_2  - \tau y_2 ) = r^4 x^2  + \tau (2p^2
- x^2  - 2z^2 ) - r^2 \tau (x_1 y_2  + x_2 y_1 ).
\end{equation}

Denote
$$
\sigma  = \tau ^2  - 2p^2 \tau  + r^4 ,\quad \chi  = \sqrt k
\geqslant 0.
$$
Then from the second relation (\ref{kh18}) the identity follows
\begin{equation}\label{kh32}
4s^2 \chi ^2  = \sigma  + 4s^2 \tau .
\end{equation}
Introduce the complex conjugate pair
$$
\mu _1  = r^2 x_1  - \tau y_1 ,\quad \mu _2  = r^2 x_2  - \tau y_2 .
$$
Eliminating the expression $x_1 y_2 + x_2 y_1$ in (\ref{kh30}),
(\ref{kh31}) with the help of (\ref{kh28}), we obtain the system
\begin{equation}\label{kh33}
\begin{array}{c}
2s(\mu _1  + \mu _2 ) = (x^2  + z^2  - \tau )^2  - 4s^2 x^2  - \sigma + 4s^2 \tau , \\
\mu _1 \mu _2  = \tau (x^2  + z^2  - \tau )^2  + \sigma x^2  - \tau
\sigma .
\end{array}
\end{equation}
Choose
$$
\lambda _1  = \sqrt {2s\mu _1  + \sigma } ,\quad \lambda _2  = \sqrt
{2s\mu _2  + \sigma }
$$
to be complex conjugate. Then system (\ref{kh33}) takes the form
$$
(\lambda _1  + \lambda _2 )^2  = \Psi _ +  (x,z),\quad (\lambda _1
- \lambda _2 )^2  = \Psi _ -  (x,z),
$$
where
$$
\Psi _ \pm  (x,z) = (x^2  + z^2  - \tau  \pm 2s\chi )^2  - 4s^2 x^2.
$$
It is solvable if
\begin{equation}\label{kh34}
\Psi _ +  (x,z) \geqslant 0,\quad \Psi _ -  (x,z) \leqslant 0.
\end{equation}

The system of inequalities (\ref{kh29}), (\ref{kh34}) defines the
region of possible motion (the RPM) in the $(x,z)$-plane. The RPM is
the projection of the integral manifold $J_{s,\tau}$. For given
$s,\tau $ the initial phase variables are algebraically expressed in
terms of $x,z$. The bifurcation diagram corresponds to the cases
when the RPM undertakes qualitative transformations as its
parameters $s,\tau$ change.

Introduce the local coordinates $s_1 ,s_2 $ in $(x,z)$-plane:
$$
s_1  = {{x^2  + z^2  + r^2 } \over {2x}},\quad s_2  = {{x^2  + z^2
- r^2 } \over {2x}}.
$$
Inequalities (\ref{kh29}) are immediately solved
\begin{equation}\label{kh35}
s_1^2  \geqslant a^2 ,\quad s_2^2  \leqslant b^2 .
\end{equation}
The corresponding region in the $(x,z)$-plane is shown in Fig.~1 for
the first quadrant. We also point out the coordinate net
$(s_1,s_2)$.

\begin{figure}[ht]\label{fig1}
\begin{center}
\includegraphics[width=8cm,keepaspectratio]{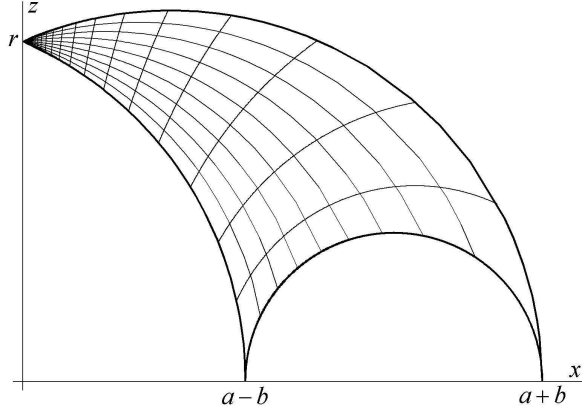}
\caption{The admissible points in the $(x,z)$-plane.}
\end{center}
\end{figure}

Let $\Pi _1 $ be the rectangle in $(s_1 ,s_2)$-plane with the
vertices $s_1 = \pm a$, $s_2  =  \pm b$. To solve system
(\ref{kh34}) express
\begin{equation}\label{kh36}
\begin{array}{c}
\displaystyle{x^2  + z^2  - \tau  = [s_1  + s_2  - {\tau  \over {r^2 }}(s_1  - s_2 )]x,}  \\
\displaystyle{\Psi _ +  (x,z) = x^2 \Lambda _ +  \Lambda _ -  ,\quad
\Psi _ -  (x,z) = x^2 {\rm M}_ +  {\rm M}_-,}
\end{array}
\end{equation}
where
$$
\begin{array}{l}
\displaystyle{\Lambda _ \pm  (s_1 ,s_2 ) = s_1  + s_2  - {{\tau  -
2s\chi}
\over {r^2 }}(s_1  - s_2 ) \pm 2s, } \\
\displaystyle{{\rm M}_ \pm  (s_1 ,s_2 ) = s_1  + s_2  - {{\tau  +
2s\chi } \over {r^2 }}(s_1  - s_2 ) \pm 2s. }
\end{array}
$$
It follows from (\ref{kh34}), (\ref{kh36}) that
\begin{equation}\label{kh37}
\Lambda _ +  (s_1 ,s_2 )\Lambda _ -  (s_1 ,s_2 ) \geqslant 0,\quad
{\rm M}_ +  (s_1 ,s_2 ){\rm M}_ -  (s_1 ,s_2 ) \leqslant 0.
\end{equation}

Consider the parallelogram $\Pi_2$ bounded by the lines $\Lambda _
\pm   = 0$, ${\rm M}_ \pm =0$. The solutions of system (\ref{kh37})
fill two half-strip regions starting at the sides of $\Pi_2$
belonging to the lines $\Lambda _ \pm   = 0$. The example of the RPM
in the $(s_1 ,s_2 )$-plane, i.e., the set of solutions of
inequalities (\ref{kh35}), (\ref{kh37}) is shown in Fig.~2.

\begin{figure}[ht]\label{fig2}
\begin{center}
\includegraphics[width=10cm,keepaspectratio]{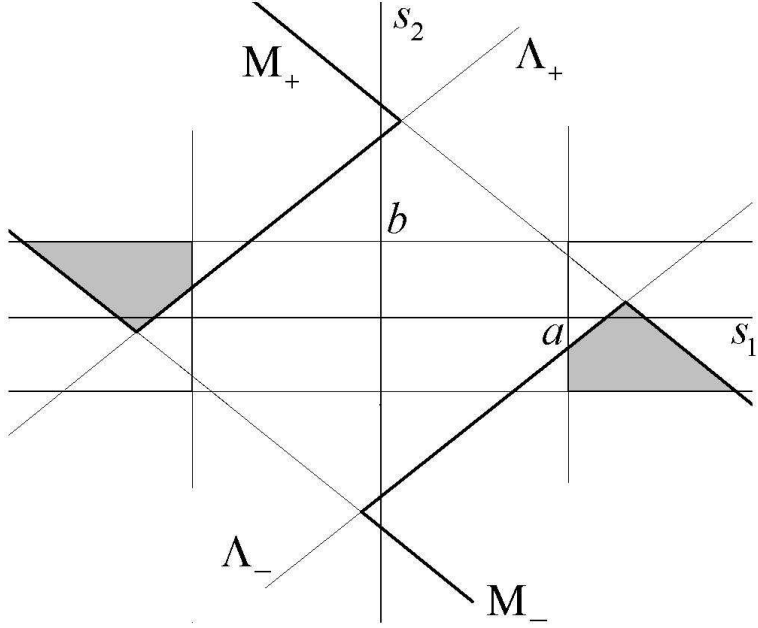}
\caption{The example of a region of possible motions $(a = 1,\;b =
0.4,\;\tau = 1.2,\;s =  - 0.6)$.}
\end{center}
\end{figure}

The further investigation is purely technical. The bifurcations of
the RPM's take place in one of the following cases: the vertex of
one parallelogram out of $\Pi_1 ,\Pi _2$ resides on the boundary of
another; the sides of the parallelograms $\Pi_1 ,\Pi_2 $ happen to
be respectively parallel (the vertices of the RPM go to the
infinity); the half stripe region degenerates and becomes a half
line. Finding all such cases gives the equations for $s,\tau$
pointed out in the theorem. Let $\Delta$ denote the set defined by
these equations in ${\bf{R}}^2 (s,\tau)$. Considering the connected
components of ${\bf{R}}^2 (s,\tau )\backslash \Delta $ we ignore
those of them which correspond to the empty RPM's. The rest of
components (the admissible regions in the integral constants space)
are shadowed in Fig.~3. The bifurcation diagram consists of those
segments of $\Delta$ that are boundaries of the admissible regions
excluding the parts of the axis $s=0$ since this value of $s$ is not
admissible by virtue of (\ref{kh18}). This way we obtain the
inequalities needed. The theorem is proved.\end{proof}

\section{On the possibility of the separation of variables }
Denote
$$
\xi  = x^2  + z^2  - \tau
$$
and consider the following second-order surface in the $(x,\xi ,\mu
)$-space
\begin{equation}\label{kh38}
\mu ^2  = \tau \xi ^2  + \sigma x^2  - \tau \sigma .
\end{equation}
Due to the second equation in (\ref{kh33}) each trajectory is
represented by some curve on this surface.

\begin{figure}[ht]\label{fig3}
\begin{center}
\includegraphics[width=8cm,keepaspectratio]{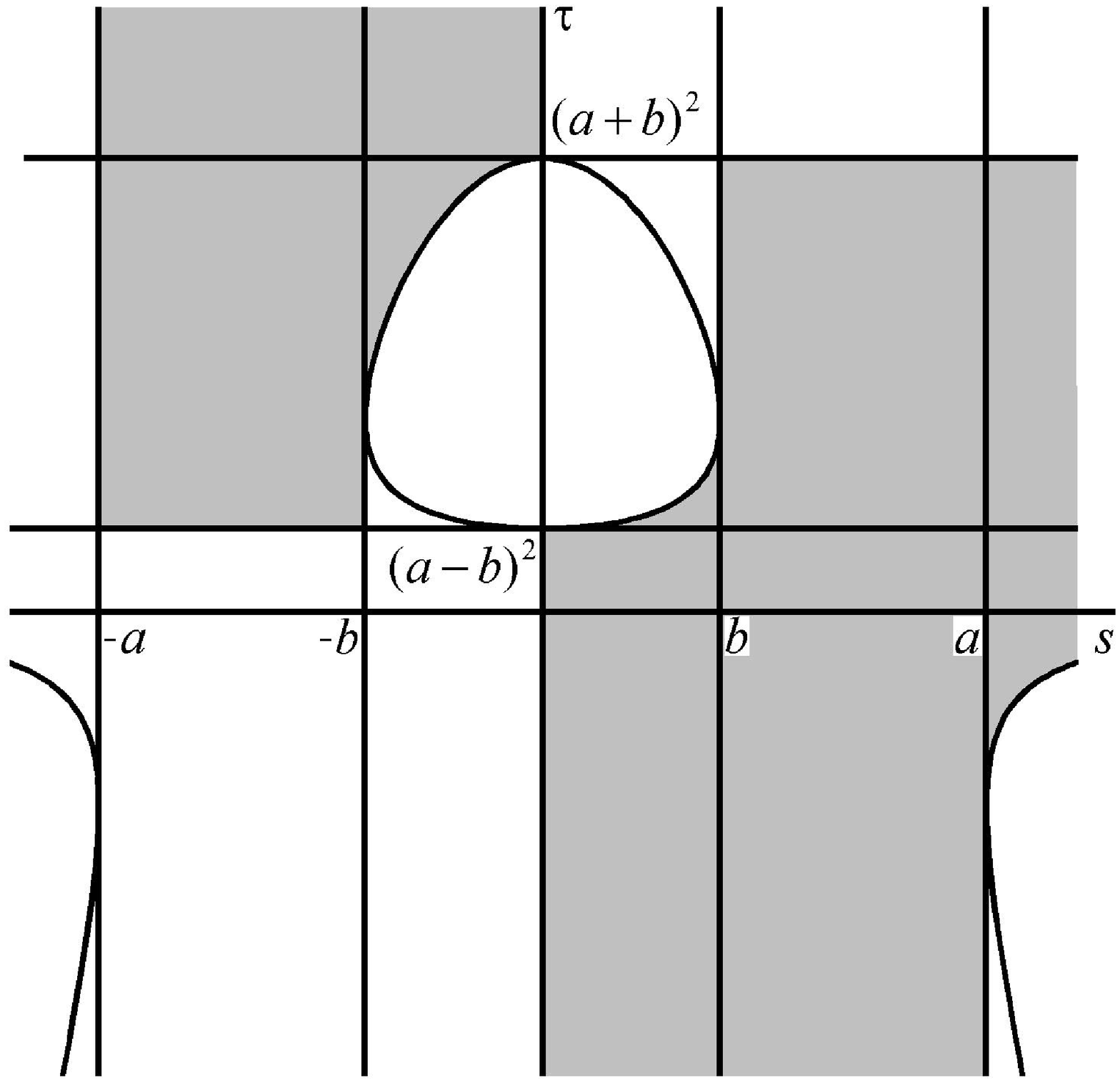}
\caption{Region of existence of motions in the $(s,\tau )$-plane.}
\end{center}
\end{figure}

Obviously, the constants $\tau ,\sigma $ cannot be simultaneously
negative. Therefore, surface (\ref{kh38}) has two families of
rectilinear generators.

The introduced constants satisfy the following two identities
\begin{equation}\label{kh39}
\sigma  + 2\tau (p^2  \pm r^2 ) = (\tau  \pm r^2 )^2 .
\end{equation}
It is easily seen that by virtue of (\ref{kh32}), (\ref{kh39}) the
equations $\Phi _ \pm   = 0$, $\Psi _ \pm = 0$ in the
$(x,\xi)$-plane define the family of lines tangent to the cross
section of surface (\ref{kh38}) by the plane $\mu = 0$. Such line is
then the projection of some generator of surface (\ref{kh38}). Since
each point on surface (\ref{kh38}) belongs exactly to two generators
the parameters of the latter can be chosen as local coordinates in
the region in the $(x,\xi)$-plane covered by surface (\ref{kh38}).

Not regarding any reality conditions, put formally
$$
\xi  = \sqrt \sigma  {{uv + 1} \over {u + v}},\quad x = \sqrt \tau
{{u - v} \over {u + v}}.
$$

After some simple transformations we get
$$
\begin{array}{l}
\displaystyle{\Phi _ +   = {1 \over {(u + v)^2 }}\varphi _1 (u)\varphi _1 (v),\quad \Phi _ -   = {1 \over {(u + v)^2 }}\varphi _2 (u)\varphi _2 (v),}  \\
\displaystyle{\Psi _ +   = {1 \over {(u + v)^2 }}\psi _1 (u)\psi _1
(v),\quad \Psi _ -   = {1 \over {(u + v)^2 }}\psi _2 (u)\psi _2
(v),}
\end{array}
$$
where
$$
\begin{array}{ll}
\varphi _1 (w) = \sqrt \sigma  (1 + w^2 ) + 2(\tau  + r^2 )w, & \varphi _2 (w) = \sqrt \sigma  (1 + w^2 ) + 2(\tau  - r^2 )w, \\
\psi _1 (w) = \sqrt \sigma  (1 + w^2 ) + 4s\chi w, & \psi _2 (w) =
\sqrt \sigma  (1 + w^2 ) - 4s\chi w.
\end{array}
$$

In the plane of the variables $u,v$ inequalities (\ref{kh29}),
(\ref{kh34}) define the set of rectangles with the sides parallel to
the coordinate axes. The fact that each connected component of any
integral manifold is represented by such a rectangle (and in the
case of a bifurcation by a segment or a pair of rectangles having a
common side) means that in these variables the equations of motion
must separate. The corresponding calculations are too long for the
restricted volume of this article and will be presented in another
publication. We only point out the connections with the above
results.

Consider the polynomial
\begin{equation}\label{kh40}
Q(w) = \varphi _1 (w)\varphi _2 (w)\psi _1 (w)\psi _2 (w)
\end{equation}
and find all the cases when it has a multiple root. The resultant of
$Q(w)$ and $Q'(w)$ in $w$ is (up to the constant multiplier)
\begin{equation}\label{kh41}
s^4 \tau ^{12} (\tau ^2  - 2p^2 \tau  + r^4 )^{14} [2s^2  - (p^2  -
r^2 )]^4 [2s^2  - (p^2  + r^2 )]^4 [\tau ^2  - 2(p^2  - 2s^2 )\tau +
r^4 ]^2 .
\end{equation}

As it was already mentioned, by virtue of equations (\ref{kh18}) at
the considered family of motions we have $s \ne 0$. The rest of
cases when expression (\ref{kh41}) vanishes lead to the equations
listed in Theorem~3. Therefore, the bifurcation diagram found above
is the part of the discriminant set of polynomial (\ref{kh40}). Such
phenomenon is also typical for the systems with algebraically
separating variables.

\end{document}